\documentclass[11pt]{article}

\usepackage[margin=2.5cm]{geometry}

\usepackage{amsmath, amssymb, amsthm, mathtools, mathrsfs}
\mathtoolsset{showonlyrefs}

\usepackage{algorithm}
\usepackage{algpseudocode}

\usepackage{graphicx}
\usepackage{subcaption}

\usepackage{siunitx}
\usepackage{enumitem}
\newlist{steps}{enumerate}{1}
\setlist[steps, 1]{label = \underline{Step \arabic*}.}

\theoremstyle{plain}
\newtheorem{theorem}{Theorem}[section]

\newtheorem{remark}{Remark}
\newtheorem{assumption}{Assumption}
\numberwithin{equation}{section}

\def\E{\mathbb{E}}

\def\R{\mathbb{R}}

\usepackage{xcolor} 

\usepackage[colorlinks=true,linkcolor=blue,citecolor=blue,urlcolor=blue]{hyperref}

\title{The Compound BSDE Method: A Fully Forward Method for Option Pricing and Optimal Stopping Problems in Finance}

\author{%
Zhipeng Huang\thanks{Mathematical Institute, Utrecht University, Postbus 80010, 3508 TA Utrecht, The Netherlands.  
Corresponding author: \texttt{z.huang1@uu.nl}} 
\and
Cornelis W. Oosterlee\thanks{Mathematical Institute, Utrecht University, Postbus 80010, 3508 TA Utrecht, The Netherlands.}
}

\date{\today} 

\usepackage{enumitem}
\usepackage{booktabs}

\newcounter{case}
\newcommand{\case}[1]{%
  \refstepcounter{case}%
  \par\medskip
  \noindent\textbf{Case \thecase: #1} 
  \par\smallskip
}

\begin{document}

\maketitle
\begin{abstract}

We propose the Compound BSDE method, a fully forward, deep-learning-based approach for solving a broad class of problems in financial mathematics, including optimal stopping. The method is based on a reformulation of option pricing problems in terms of a system of backward stochastic differential equations (BSDEs), which offers a new perspective on the numerical treatment of compound options and optimal stopping problems such as Bermudan option pricing. Building on the classical deep BSDE method for a single BSDE, we develop an algorithm for compound BSDEs and establish its convergence properties. In particular, we derive an \emph{a posteriori} error estimate for the proposed method. Numerical experiments demonstrate the accuracy and computational efficiency of the approach, and illustrate its effectiveness for high-dimensional option pricing and optimal stopping problems.

\end{abstract}


\tableofcontents

\section{Introduction}\label{sec:intro}

Option pricing is a central topic in financial mathematics, especially for derivatives with complex payoff structures, early-exercise features, or high-dimensional underlying dynamics. Over the past decades, a wide range of numerical methods has been developed to address these challenges. Without attempting to provide an exhaustive overview, we mention partial differential equation (PDE) based methods \cite{wong2008artificial,forsyth2002quadratic,reisinger2012use}, Fourier-based techniques such as the COS method \cite{fang2009novel,fang2009pricing}, and Monte Carlo regression-based approaches \cite{longstaff2001valuing,bouchard2012monte}.

In addition to these approaches, backward stochastic differential equations (BSDEs) \cite{el1997backward} provide a flexible framework for option pricing. Via the Feynman--Kac representation \cite{pardouxpeng}, a pricing problem can be reformulated as a BSDE and solved numerically. A broad class of numerical schemes for BSDEs and forward--backward SDEs (FBSDEs) has been developed and analyzed; see, for example, \cite{bouchard2004discrete,bouchard2004malliavin,gobet2005regression,PagesPhamPrintems2004,bender2007forward,bender2008time}. More recently, machine-learning-based approaches have attracted considerable interest due to their ability to address high-dimensional problems, including the deep BSDE method \cite{han2018,hanlong2020}, deep splitting \cite{BeckBeckerCheridito2019}, and deep backward dynamic programming \cite{hure2020deep}, as well as related variants \cite{andersson2023,ji2022,negyesi2024generalized}.

Reflected BSDEs provide a natural framework for pricing American- and Bermudan-style options; see, for example, \cite{gobet2008numerical, JFCr2012, memin2008convergence, BouchardChassagneux2008, hure2020deep, negyesi2025deep} for numerical methods for solving this type of BSDE. These approaches share a similar philosophy with the least-squares Monte Carlo method \cite{longstaff2001valuing}: one simulates the underlying dynamics forward in time and then optimizes backward. In recent contributions \cite{wang2018deep, gao2023convergence}, a backward variant of the deep BSDE method has been studied, in which the loss functional is replaced by the variance of the $Y$ process at the initial time. Numerical results suggest that this approach may be applicable to Bermudan option pricing. However, the convergence of this method for pricing early-exercise options has not been established in \cite{gao2023convergence} and remains an open problem.

In this paper, we introduce the Compound BSDE method, a deep-learning-based approach for solving a broad class of option pricing and optimal stopping problems in finance. The main contributions of this work can be summarized as follows:
\begin{itemize}
    \item We introduce the compound BSDE, a new BSDE formulation with a wide range of financial applications, in particular for option pricing and optimal stopping. The formulation is inspired by the payoff structure of compound options. A detailed discussion of this class of exotic options and their connection to our framework is provided in Section~\ref{sec4}.
    \item Building on the deep BSDE method, we develop the Compound BSDE method tailored to the new formulation and derive an \emph{a posteriori} error estimate. We emphasize that the resulting algorithm is naturally a \emph{fully forward method} and can be applied directly to optimal stopping problems, in contrast to many existing approaches in the aforementioned literature.
    \item We present numerical experiments that validate the theoretical results. The results demonstrate that the proposed method achieves accurate prices and hedges for a variety of option types, including Bermudan-style options in high-dimensional settings.
\end{itemize}

The remainder of the paper is organized as follows. In Section~\ref{sec2}, we introduce the compound BSDE formulation and the associated numerical method. In Section~\ref{sec3}, we discuss well-posedness of the formulation and establish convergence of the proposed method. Section~\ref{sec4} explains how the compound BSDE framework can be applied to various option pricing and optimal stopping problems. Numerical results are presented in Section~\ref{sec5}. Section~\ref{sec6} concludes the paper.

\section{The Compound BSDE Method}\label{sec2}

In this section, we introduce the compound BSDE formulation and the associated numerical method. We begin with a brief review of the deep BSDE method, which serves as an important building block for our approach, and then define the compound BSDE motivated by the structure of compound options.

\subsection{Review of the deep BSDE method}

A backward stochastic differential equation (BSDE) on a finite time horizon $[0,T]$ is formulated on a filtered probability space
$(\Omega,\mathcal{F},\{\mathcal{F}_t\}_{0\le t\le T},\mathbb{P})$, and typically takes the form
\begin{equation}\label{eq:BSDE}
\left\{
\begin{aligned}
\mathcal{X}_t & = \mathcal{X}_0
+ \int_0^t \mu(s,\mathcal{X}_s)\,\mathrm{d}s
+ \int_0^t \sigma(s,\mathcal{X}_s)\,\mathrm{d}W_s, \\
\mathcal{Y}_t & = g(\mathcal{X}_T)
+ \int_t^{T} f(s,\mathcal{X}_s,\mathcal{Y}_s,\mathcal{Z}_s)\,\mathrm{d}s
- \int_t^{T} \mathcal{Z}_s\,\mathrm{d}W_s ,
\end{aligned}
\right.
\end{equation}
where $W \coloneqq \{W_t\}_{0\le t\le T}$ is a Brownian motion, and $\mu$, $\sigma$, $g$, and $f$ are deterministic functions. The solution of \eqref{eq:BSDE} is a triple of adapted processes
$(\mathcal{X},\mathcal{Y},\mathcal{Z}) \coloneqq \{(\mathcal{X}_t,\mathcal{Y}_t,\mathcal{Z}_t)\}_{0\le t\le T}$ satisfying appropriate integrability conditions.

The deep BSDE method \cite{han2018} approximates the solution of \eqref{eq:BSDE} by solving a stochastic optimization problem based on a time discretization. Specifically, it considers
\begin{subequations}\label{eq:deepBSDE}
\begin{align}
& \inf_{\mathcal{Y}_{t_0}^\pi, \{\mathcal{Z}_{t_i}^\pi\}  }
\quad \mathbb{E}\big\|g(\mathcal{X}_T^\pi)-\mathcal{Y}_T^\pi\big\|^2,
\label{eq:deepBSDE-a}\\
& \text{subject to}
\left\{
\begin{aligned}
\mathcal{X}_{t_{i+1}}^\pi &=
\mathcal{X}_{t_i}^\pi
+ b(t_i,\mathcal{X}_{t_i}^\pi)h
+ \sigma(t_i,\mathcal{X}_{t_i}^\pi)\Delta W_{t_i},\\[3pt]
\mathcal{Y}_{t_{i+1}}^\pi &=
\mathcal{Y}_{t_i}^\pi
- f(t_i,\mathcal{X}_{t_i}^\pi,\mathcal{Y}_{t_i}^\pi,\mathcal{Z}_{t_i}^\pi)h
+ \mathcal{Z}_{t_i}^\pi \Delta W_{t_i} ,
\end{aligned}
\right.
\label{eq:deepBSDE-b}
\end{align}
\end{subequations}
where we denote $\| x \|$  the Euclidean norm if $x$ is a vector or the Frobenius norm if $x$ is a matrix,  $\Delta W_{t_i} \coloneqq W_{t_{i+1}} - W_{t_{i}}$ is the Brownian increment, $\pi$ denotes a uniform partition of $[0, T]$ with step size $h$. The quantities $\mathcal{Y}_{t_0}^\pi$ and $\mathcal{Z}_{t_i}^\pi$ are parameterized by neural networks and serve as approximations of $\mathcal{Y}_{t_0}$ and $\mathcal{Z}_{t_i}$, respectively. The objective functional \eqref{eq:deepBSDE-a} acts as the loss function during the training of neural networks, and enforces the terminal condition $\mathcal{Y}_T=g(\mathcal{X}_T)$.

Two properties of \eqref{eq:deepBSDE} are particularly relevant for our development. First, as noted in \cite{hanlong2020}, the initial state $\mathcal{X}_{t_0}^\pi$ need not be deterministic, but may be any square-integrable random variable. In this case, $\mathcal{Y}_{t_0}^\pi$ learns the dependence on the distribution of $\mathcal{X}_{t_0}^\pi$ rather than a fixed constant. Second, convergence of the method depends primarily on how small the objective function value is, and does not depend explicitly on $\mathcal{Y}_{t_0}^\pi$. We remark that when $\mathcal{X}_{t_0}^\pi$ is random, the backward deep BSDE method of \cite{gao2023convergence} cannot be applied directly, since the property $\operatorname{Var}(\mathcal{Y}_0)=0$ no longer holds.

\subsection{Motivation: compound options}

We now explain the intuition behind the compound BSDE formulation. Consider a simple compound option, namely a call-on-call option. This product can be viewed as a call option (the outer option) whose underlying asset is another call option (the inner option). When the outer option expires at time $T_1>0$, the holder has the right to purchase the inner option, which itself expires at a later time $T_2>T_1$, at a specified strike price.

In the BSDE framework \eqref{eq:BSDE}, the $\mathcal{Y}$ process represents the option value. In the current call-on-call setting, this naturally leads to two BSDEs: one describing the value of the inner option and one describing the value of the outer option, coupled at the intermediate time $T_1$ through the exercise decision. This idea extends naturally to an $M$-fold compound option, where multiple BSDEs are linked at a sequence of intermediate times $T_j$ for $j=1,2,\dots M-1$.

\subsection{Definition of the compound BSDE}

We now introduce the compound BSDE formally. Let $(\Omega,\mathcal{F},\{\mathcal{F}_t\}_{0\le t\le T},\mathbb{P})$ be a filtered probability space, where $\{\mathcal{F}_t\}$ is the natural filtration of a $d_3$-dimensional Brownian motion $W=\{W_t\}_{0\le t\le T}$. Let $M\in\mathbb{N}$ and $0=T_0<T_1<\cdots<T_M=T$. The $M$-fold compound BSDE on $[0,T]$ is defined by
\begin{equation}\label{eq:compoundBSDE}
\left\{
\begin{aligned}
X_t &=
X_0 + \int_0^t \mu(s,X_s)\,\mathrm{d}s
+ \int_0^t \sigma(s,X_s)\,\mathrm{d}W_s,  \quad \forall t\in[0, T] \\
Y_{1,t} &=
g_1(X_{T_1},Y_{2,T_1})
+ \int_t^{T_1} f_1(s,X_s,Y_{1,s},Z_{1,s})\,\mathrm{d}s
- \int_t^{T_1} Z_{1,s}\,\mathrm{d}W_s, \quad \forall t\in[T_0, T_1]  \\ 
Y_{2,t} &=
g_2(X_{T_2},Y_{3,T_2})
+ \int_t^{T_2} f_2(s,X_s,Y_{2,s},Z_{2,s})\,\mathrm{d}s
- \int_t^{T_2} Z_{2,s}\,\mathrm{d}W_s, \quad \forall t\in[T_1, T_2] \\
&\hspace{1.2cm}\vdots \\ 
Y_{M,t} &=
g_M(X_T)
+ \int_t^{T_M} f_M(s,X_s,Y_{M,s},Z_{M,s})\,\mathrm{d}s
- \int_t^{T_M} Z_{M,s}\,\mathrm{d}W_s, \quad \forall t\in[T_{M-1}, T_M] .
\end{aligned}
\right.
\end{equation}

Here, $T_j$ and $g_j$ for $j=1,\ldots, M-1$ are referred to as the \emph{compounding times} and \emph{compounding conditions}, respectively, to distinguish them from the terminal time $T_M$ and terminal condition $g_M$. The functions $g_j$, $f_j$, $\mu$, and $\sigma$ are all deterministic. The solution to \eqref{eq:compoundBSDE}, denoted by $(X,Y,Z)$, consists of the $\mathbb{R}^{d_1}$-valued forward process $X$ and the family of $\mathbb{R}^{d_2} \times \mathbb{R}^{d_2\times d_3}$ - valued backward processes $\{(Y_{j, t},Z_{j,t})\}_{j, t}$, which are adapted and square-integrable. The well-posedness of this system will be discussed in Section~\ref{sec3}.

The compound BSDE \eqref{eq:compoundBSDE} comprises a single forward SDE coupled with $M$ backward equations defined on successive time intervals. These equations interact only at the compounding times through the functions $g_j$, which may depend on the unknown values $Y_{j+1,T_j}$. This distinguishes the compound BSDE from a collection of independent BSDEs and requires the system to be treated as a whole. This coupling or compounding structure plays a central role in both the convergence analysis in Section~\ref{sec3} and the applications discussed in Section~\ref{sec4}.

\subsection{The Compound BSDE method}

To solve \eqref{eq:compoundBSDE} numerically, we rewrite all backward equations in a forward simulation framework and enforce the compounding and terminal conditions by minimizing a joint objective functional. This leads to the Compound BSDE method,
\begin{align}\label{method:compoundBSDE}
& \inf_{ \{Y_{j,T_{j-1}}^\pi\},\ \{Z_{j,t_i}^\pi\} }
\quad \sum_{j=1}^{M-1} \mathbb{E} \|g_j(X_{T_j}^\pi, Y_{j+1, T_j}^\pi) - Y_{j, T_j}^\pi\|^2 + \mathbb{E} \|g_M(X_{T_M}^\pi) - Y_{M, T_M}^\pi\|^2  \\
& \text{s.t.  }  
\left\{
\begin{aligned}
X_{t_{i+1}}^\pi & = X_{t_i}^\pi+b\left(t_i, X_{t_i}^\pi\right) h+\sigma\left(t_i, X_{t_i}^\pi\right) \Delta W_{t_i},  \quad X_{t_0}^\pi = X_0 \\ 
Y_{1, t_{i+1}}^\pi & = Y_{1, t_i}^\pi - f_1(t_i, X_{t_i}^\pi, Y_{1, t_i}^\pi, Z_{1, t_i}^\pi ) h + Z_{1, t_i}^\pi \Delta  W_{t_i}, 
\\
Y_{2, t_{i+1}}^\pi & = Y_{2, t_i}^\pi - f_2(t_i, X_{t_i}^\pi, Y_{2, t_i}^\pi, Z_{2, t_i}^\pi ) h + Z_{2, t_i}^\pi \Delta  W_{t_i},    \\
    & \vdotswithin{=}       \\
Y_{M, t_{i+1}}^\pi & = Y_{M, t_i}^\pi - f_M(t_i, X_{t_i}^\pi, Y_{M, t_i}^\pi, Z_{M, t_i}^\pi ) h + Z_{M, t_i}^\pi \Delta  W_{t_i},   
\end{aligned}
\right.
\end{align}

Here, we adopt the same conventions for the norms, the Brownian increments $\Delta W_{t_i}$ and the uniform partition $\pi$ as in the description of the deep BSDE method. The initial values $\{Y_{j,T_{j-1}}^\pi\}_j$ and the processes $\{Z_{j,t_i}^\pi\}_{j,t_i}$ are parameterized by neural networks. Moreover, let $N_j$ denote the number of time discretization steps for the process $Y_{j,t_i}^\pi$ running on $[T_{j-1}, T_j]$, for $j=1,2,\ldots,M$, so that $N_1 + N_2 + \cdots + N_M = N$ is the total number of time steps. For convenience, we assume that $(T_j - T_{j-1})/N_j = h$ for all $j$.

The resulting algorithm is a fully forward method: both the forward SDE and all backward equations are simulated forward in time. In contrast to the classical deep BSDE method, the first $M-1$ compounding conditions specified by $g_j$ are not known, as all these $g_j$ may depend on the unknown solution $Y$ itself. Consequently, all $M$ BSDEs must be learned \textit{simultaneously}. This simultaneous enforcement of all compounding conditions is the defining feature of the Compound BSDE method.

\section{Convergence Analysis}\label{sec3}

In this section, we first address the well-posedness of the compound BSDE formulation~\eqref{eq:compoundBSDE} and collect the $L^2$-regularity properties needed in the convergence analysis. Combining these results with the \emph{a posteriori} error estimate for the deep BSDE method applied to~\eqref{eq:BSDE}, we then obtain a concise convergence proof for the Compound BSDE method~\eqref{method:compoundBSDE}.

\begin{assumption}\label{assum1}
The coefficients in~\eqref{eq:compoundBSDE} satisfy the following conditions.
\begin{enumerate}[label=(\roman*).]
\item
The mappings
$b:[0,T]\times\R^{d_1}\to\R^{d_1}$,
$\sigma:[0,T]\times\R^{d_1}\to\R^{d_1\times d_3}$,
$f_j:[T_{j-1},T_j]\times\R^{d_1}\times\R^{d_2}\times\R^{d_2\times d_3}\to\R^{d_2}$ for all $j$,
$g_M:\R^{d_1}\to\R^{d_2}$, and
$g_j:\R^{d_1}\times\R^{d_2} \to\R^{d_2}$ for $1\leq j \leq M-1$,
are all deterministic.

\item
The functions $b(\cdot,0)$ and $\sigma(\cdot,0)$ are bounded. Moreover, $f_j(\cdot,0,0,0)$ for all $j$, $g_j(0,0)$ for $1\leq j \leq M-1$, and $g_M(0)$ are bounded.

\item
The functions $b$, $\sigma$, $f_j$, and $g_j$ for all $j$, are uniformly Lipschitz continuous in $(x,y,z)$.

\item
The functions $b$, $\sigma$, and $f_j$ for all $j$, are uniformly $\tfrac12$-H\"older continuous in $t$.

\end{enumerate}
\end{assumption}

\begin{theorem}\label{thm:pde}
Let Assumption \ref{assum1} hold. Then we have the following.
\begin{enumerate}[label=(\roman*).]
\item
The compound BSDE~\eqref{eq:compoundBSDE} admits a unique adapted solution $(X, Y, Z) $.

\item
For each $j=1,\dots,M$, there exists a mapping $u_j$ such that
\begin{equation}
Y_{j,t}=u_j(t,X_t),\qquad t\in[T_{j-1},T_j],
\end{equation}
and $u_j$ is the viscosity solution to the system of semilinear PDEs, for $\ell=1,2,\cdots, d_2$,
\begin{equation}\label{eq:systemofpde}
\left\{
\begin{aligned}
\partial_t u_j^\ell(t, x)
& + \frac12 \operatorname{Tr} \Bigl(\sigma\sigma^{\mathrm{T}}(t,x) \partial_x^2 u_j^\ell (t, x) \Bigr)
+ b^{\mathrm{T}}(t,x) \partial_x u_j^\ell (t,x) \\
&\qquad + f_j^\ell \Bigl(t, x, u_j(t,x), (\partial_x u_j(t,x))^{\mathrm{T}} \sigma(t,x) \Bigr) = 0,
\qquad \forall (t,x)\in[T_{j-1}, T_j) \times\R^{d_1},\\
u_j(T_j,x)& = \bar g_j(x),   \qquad \forall x \in \R^{d_1},
\end{aligned}
\right.
\end{equation}
where $\bar g_M(x)\coloneqq g_M(x)$ and, for $1\le j\le M-1$,
\begin{equation}
\bar g_j(x) \coloneqq g_j \bigl(x, u_{j+1}(T_j,x) \bigr).
\end{equation}

Moreover, each $u_j$ is uniformly Lipschitz in $x$, i.e.,
\begin{equation}
\|u_j(t,x_1)-u_j(t,x_2)\| \leq C\|x_1-x_2\|
\qquad \text{for all } t\in[T_{j-1},T_j],\ x_1,x_2\in\R^{d_1}.
\end{equation}

\item If we additionally assume that $u_j \in C^{1,2} ( [T_{j-1}, T_j] \times \mathbb{R}^{d_1}, \mathbb{R}^{d_2})$, then the PDE \eqref{eq:systemofpde} holds in the classical sense, and we have, for all $j$,
\begin{equation}\label{eq:feynman-kac}
Y_{j,t}=u_j(t,X_t), \quad Z_{j,t} = (\nabla_x u_j(t, X_t))^{\mathrm{T}} \sigma(t, X_t), \quad \forall t\in [T_{j-1}, T_j]
\end{equation}

\end{enumerate}
\end{theorem}

\begin{proof}
We argue backwards in time, starting from $j=M$. For $j=M$, the terminal condition is $\bar g_M=g_M$, which is Lipschitz by Assumption~\ref{assum1}. Existence and uniqueness of the BSDE on $[T_{M-1},T_M]$, as well as the Markovian representation $Y_{M,t}=u_M(t,X_t)$ with $u_M$ a viscosity solution that is Lipschitz in $x$, follow from standard BSDE theory; see, for instance, Theorems~5.2.1 and~5.5.8 in \cite{zhang2017backward}.

Assume now that the claim holds for $u_{j+1}$. By the Lipschitz continuity of $g_j$ in its arguments and the Lipschitz property of $u_{j+1}(T_j,\cdot)$, the induced function $\bar g_j(\cdot)$ is again Lipschitz. Applying the same BSDE well-posedness and Markovian representation result on the interval $[T_{j-1},T_j]$ yields existence and uniqueness for $(Y_j,Z_j)$ and the corresponding PDE representation. Iterating this argument for $j=M-1,M-2,\dots,1$ completes the proof.
\end{proof}

Using the well-posedness results, we can further derive \(L^2\)-regularity estimates for the solution \((X,Y,Z)\) of \eqref{eq:compoundBSDE} based on classical BSDE theory \cite{zhang2004numerical, zhang2017backward}. To simplify the subsequent analysis, we impose Assumption \ref{assume: cadlag}, which ensures that the process \(Z\) admits a càdlàg modification. For notational convenience, recall the uniform partition $\pi$, we define the projection map $\pi(t) = t_i$ for $t \in [t_i, t_{i+1})$.

\begin{assumption}\label{assume: cadlag}
The diffusion coefficient $\sigma(t,x)$ in \eqref{eq:compoundBSDE} satisfies one of the following conditions: either $\sigma(t,x)$ is uniformly Lipschitz continuous in $t$, or dimensions $d_1 = d_3$ and the uniform ellipticity condition $\sigma(t,x)\sigma(t,x)^\top \ge \delta \mathbf{I}_{d_1}$ holds for some constant $\delta>0$.
\end{assumption}

\begin{theorem}[$L^2$-regularity]\label{thm:L2-regularity}
Let Assumptions \ref{assum1} and \ref{assume: cadlag} hold, and \((X,Y,Z)\) be the solution to the compound BSDE \eqref{eq:compoundBSDE}. For each \(j=1,2,\dots,M\), the \(j\)-th BSDE, together with the associated forward process \(X_t\) on \([T_{j-1},T_j]\), satisfies the following regularity estimate,
\begin{equation}
\sup_{t \in [T_{j-1}, T_j]}
\mathbb{E}\!\left[
\left\|X_t - X_{\pi(t)}\right\|^2
+
\left\|Y_{j,t} - Y_{j,\pi(t)}\right\|^2
\right]
+
\mathbb{E}\!\left[
\int_{T_{j-1}}^{T_j}
\left\|Z_{j,t} - Z_{j,\pi(t)}\right\|^2 \,\mathrm{d}t
\right]
\leq
C h, 
\end{equation}
where $C$ is a constant independent of the time step size $h$.
\end{theorem}

In what follows, we first present a proof for the deep BSDE method for equations of the form \eqref{eq:BSDE}, as this will serve as the main building block for establishing our Compound BSDE framework. We note that rigorous analyses of related methods have already been carried out in more complex settings, e.g. coupled forward–backward stochastic differential equations \cite{hanlong2020, negyesi2024generalized}, BSDEs with jumps \cite{bsdejumps}, BSDEs with non-Lipschitz coefficients \cite{jiang2021}, and mean-field type equations \cite{reisinger2024posteriori}. In contrast, our objective here is to provide a streamlined and self-contained proof for the basic case. This allows us to emphasize the essential ideas of our Compound BSDE method, while avoiding many of the technical conditions and complications that arise in the complex setting, particularly those associated with an implicit scheme for the BSDE.

The following assumption forms part of the sufficient conditions for convergence, and we shall further discuss it in a later remark. We also recall that the deep BSDE method \eqref{eq:deepBSDE} is a special case of the Compound BSDE method \eqref{eq:compoundBSDE} with \(M=1\), where all the subscripts \(j\) are omitted as they are not needed for $M=1$. This observation clarifies the use of the assumptions in Theorem \ref{thm:deepBSDE} below.

\begin{assumption}\label{assume:Ky}
The Lipschitz constants of the driver \(f_j\) with respect to \(y\) and \(z\) are sufficiently small.
\end{assumption}

\begin{theorem}[A posteriori error estimate for the deep BSDE method]\label{thm:deepBSDE}
Let Assumptions \ref{assum1}, \ref{assume: cadlag}, and \ref{assume:Ky} hold for the BSDE \eqref{eq:BSDE}. Let $(\mathcal{X},\mathcal{Y},\mathcal{Z})$ denote its solution, and let $(\mathcal{X}^\pi, \mathcal{Y}^\pi, \mathcal{Z}^\pi)$ generated by the deep BSDE method \eqref{eq:deepBSDE}. Then the method satisfies the following \emph{a posteriori} error estimate,
\begin{equation}\label{estimate:deepBSDE}
\sup_{t \in [0,T]}
\mathbb{E} \Big[ \|\mathcal{X}_t - \mathcal{X}_{\pi(t)}^\pi\|^2
+ \|\mathcal{Y}_t - \mathcal{Y}_{\pi(t)}^\pi\|^2
\Big]
+
\int_0^T \mathbb{E} \Big[  \|\mathcal{Z}_t - \mathcal{Z}_{\pi(t)}^\pi\|^2 \Big] \,\mathrm{d}t  
\leq 
C \Big( h + \mathbb{E} \| g(\mathcal{X}_T) - \mathcal{Y}_T^\pi \|^2 \Big),
\end{equation}
where $C$ is a constant independent of the step size $h$.

If there is no direct access to \(g(\mathcal{X}_T)\), it can be replaced by \(g\!\left(\mathcal{X}_T^\pi\right)\) on the right-hand side of \eqref{estimate:deepBSDE}.

\end{theorem}

\begin{proof}

For \(0 \le s \le T\), we introduce the following notation,
\begin{equation}
\Delta \mathcal{X}_s \coloneqq \mathcal{X}_s - \mathcal{X}_{\pi(s)}^\pi,\qquad
\Delta \mathcal{Y}_s \coloneqq \mathcal{Y}_s - \mathcal{Y}_{\pi(s)}^\pi,\qquad
\Delta \mathcal{Z}_s \coloneqq \mathcal{Z}_s - \mathcal{Z}_{\pi(s)}^\pi,
\end{equation}
\begin{equation}
\Delta f_s \coloneqq
f(s, \mathcal{X}_s, \mathcal{Y}_s, \mathcal{Z}_s)
-
f\bigl(\pi(s), \mathcal{X}_{\pi(s)}^\pi, \mathcal{Y}_{\pi(s)}^\pi, \mathcal{Z}_{\pi(s)}^\pi\bigr),
\end{equation}
where we recall that \(\pi(s)=t_i\) for \(s \in [t_i, t_{i+1})\).

Let $C$ denote a generic constant that may change from line to line in the rest of the proof, and let \(K_x\), \(K_y\), and \(K_z\) denote the Lipschitz constants of \(f\) with respect to \(x\), \(y\), and \(z\). Applying the Cauchy--Schwarz and Hölder inequalities, together with the Lipschitz continuity of the driver \(f\), we obtain the following useful estimate related to \(\Delta f_s\). For \(0 \le i \le N\),
\begin{equation}\label{estimate:Deltaf}
\begin{aligned}
\E \left\| \int_{t_0}^{t_i} \Delta f_s ds \right\|^2 
& \leq  T \E  \int_{t_0}^{T} \left( h + K_x^2 \|\Delta \mathcal{X}_s\|^2 + K_y^2 \|\Delta \mathcal{Y}_s\|^2 + K_z^2 \|\Delta \mathcal{Z}_s\|^2   \right) d s \\
& \leq  T^2 h + C T^2 K_x^2 h +   T K_y^2 \E \left[ \int_{t_0}^{T} \|\Delta \mathcal{Y}_s\|^2 ds \right] +  T K_z^2 \E  \left[ \int_{t_0}^{T} \|\Delta \mathcal{Z}_s\|^2   d s  \right] \\
& \leq  C h  +  T K_z^2 \E  \left[ \int_{t_0}^{T} \|\Delta \mathcal{Z}_s\|^2 ds \right] +  T K_y^2  \E \left[ \int_{t_0}^{T} 2 \|\mathcal{Y}_s - \mathcal{Y}_{\pi(s)} \|^2 + 2 \|\mathcal{Y}_{\pi(s)} - \mathcal{Y}_{\pi(s)}^\pi \|^2  ds  \right]   \\
& \leq  C h  +  T K_z^2 \E  \left[ \int_{t_0}^{T} \|\Delta \mathcal{Z}_s\|^2 ds \right] +  2 T K_y^2  \E \left[ \int_{t_0}^{T}  \|\mathcal{Y}_{\pi(s)} - \mathcal{Y}_{\pi(s)}^\pi \|^2  ds  \right]   \\
& \leq  C h  +  T K_z^2 \E  \left[ \int_{t_0}^{T} \|\Delta \mathcal{Z}_s\|^2 ds \right] +  2 T^2 K_y^2   \max_{0\leq i\leq N-1}  \E \left[ \|\mathcal{Y}_{t_i} - \mathcal{Y}_{t_i}^\pi \|^2  \right] .
\end{aligned}
\end{equation}
where we used Theorem \ref{thm:L2-regularity} for the \(L^2\)-regularity of \(\mathcal{X}\) and \(\mathcal{Y}\) in the second and forth inequalities, respectively.

Consequently, using It\^o isometry and \eqref{estimate:Deltaf}, for $0\leq i \leq N$ we have,
\begin{equation}
\begin{aligned}
\E \|\Delta \mathcal{Y}_{t_i}\|^2  
& = \E \left\| \Delta \mathcal{Y}_{t_0} - \int_{t_0}^{t_i} \Delta f_s \mathrm{d}s + \int_{t_0}^{t_i} \Delta \mathcal{Z}_s \mathrm{d}W_s \right\|^2 \\
& \leq 3  \E \|\Delta \mathcal{Y}_{t_0} \|^2 +  3 \E \left\| \int_{t_0}^{t_i} \Delta f_s \mathrm{d}s \right\|^2 + 3 \E \left\| \int_{t_0}^{t_i} \Delta \mathcal{Z}_{s}  \mathrm{d}W_s \right\|^2  \\
& \leq 3 \E \|\Delta \mathcal{Y}_{t_0} \|^2 
+  3\left( C h  +  T K_z^2 \E  \left[ \int_{t_0}^{T} \|\Delta \mathcal{Z}_s\|^2 ds \right] +  2 T^2 K_y^2  \max_{0\leq i\leq N-1}  \E \|\Delta \mathcal{Y}_{t_i} \|^2  \right) 
+ 3 \E \left[ \int_{t_0}^{t_i} \|\Delta \mathcal{Z}_{s}\|^2  \mathrm{d} s \right]   \\ 
& = Ch + 3 \E \|\Delta \mathcal{Y}_{t_0} \|^2 + (3 TK_z^2 + 3) \E \left[ \int_{t_0}^{T} \|\Delta \mathcal{Z}_{s}\|^2  \mathrm{d} s \right] + 6T^2 K_y^2  \max_{0\leq i\leq N-1}  \E \|\Delta \mathcal{Y}_{t_i} \|^2 .
\end{aligned}
\end{equation}
Taking the maximum over \(0 \le i \le N-1\) on the left-hand side of the above inequality, and assuming that \(6 T^2 K_y^2 < 1\), which is ensured when Assumption~\ref{assume:Ky} holds to a sufficient extent, we can rearrange the terms to obtain 
\begin{equation}\label{estimate:maxY}
\max_{0 \leq i \leq N-1} \E \left\|\Delta \mathcal{Y}_{t_i} \right\|^2
\leq  C h + \frac{3}{1-6T^2 K_y^2 } \E \left\| \Delta \mathcal{Y}_{t_0} \right\|^2  + 
\frac{3\left(1+T K_z^2 \right)}{1-6 T^2 K_y^2 } \E \left[ \int_0^T \left\|\Delta \mathcal{Z}_s \right\|^2 \mathrm{d}s \right] .
\end{equation}
It now remains to bound the second and third terms on the right-hand side of \eqref{estimate:maxY} in terms of the objective functional
\(\mathbb{E}\bigl \| g(\mathcal{X}_T) - \mathcal{Y}_T^\pi \bigr\|^2\) used in the deep BSDE method \eqref{eq:deepBSDE}. Recall that $g(\mathcal{X}_T) = \mathcal{Y}_T$, then we can derive
\begin{equation}\label{estimate:lossfun}
\begin{aligned}
& \E \left\| g(\mathcal{X}_T) - \mathcal{Y}_{T}^\pi \right\|^2 \\
= & \E \left\| \Delta \mathcal{Y}_{t_0} - \int_{t_0}^{T} \Delta f_s \mathrm{d}s + \int_{t_0}^{T} \Delta \mathcal{Z}_s  \mathrm{d}W_s \right\|^2  \\
= & \E \left\| \Delta \mathcal{Y}_{t_0} - \int_{t_0}^{T} \Delta f_s \mathrm{d}s \right\|^2 
+ 2 \E  \left[ \left( \Delta \mathcal{Y}_{t_0} - \int_{t_0}^{T} \Delta f_s \mathrm{d}s  \right)^\top \left( \int_{t_0}^{T} \Delta \mathcal{Z}_s \mathrm{d}W_s \right)  \right] 
+ \E \left\| \int_{t_0}^{T} \Delta \mathcal{Z}_s \mathrm{d}W_s \right\|^2  \\
= & \E \left\| \Delta \mathcal{Y}_{t_0} - \int_{t_0}^{T} \Delta f_s \mathrm{d}s \right\|^2 
- 2 \E \left[ \left(\int_{t_0}^T \Delta f_s \mathrm{d}s \right)^\top \left(\int_{t_0}^T \Delta \mathcal{Z}_s \mathrm{d}W_s \right)   \right]  + \E \left[\int_{t_0}^{T} \|\Delta \mathcal{Z}_s\|^2 \mathrm{d}s  \right]\\
\geq & (1 - \epsilon) \E \left\| \Delta \mathcal{Y}_{t_0}  \right\|^2 + (1 - \epsilon^{-1} - \delta^{-1}) \E \left\| \int_{t_0}^{T} \Delta f_s \mathrm{d}s \right\|^2 
+ (1-\delta) \E \left[ \int_{t_0}^{T} \left\| \Delta \mathcal{Z}_s \right\|^2 \mathrm{d}s  \right] ,
\end{aligned}
\end{equation}
where we use the fact that $\Delta \mathcal{Y}_{t_0}$ is independent with $W_s$ for $s\geq 0$, and the constants $\epsilon>0$ and $\delta>0$ arise from the application of Young's inequality.

Since our goal is to bound the terms $\E \left\| \Delta \mathcal{Y}_{t_0} \right\|^2$ and $\E \int_{t_0}^T \|\Delta \mathcal{Z}_s\|^2 \mathrm{d}s$ in terms of $\E \left\|g(\mathcal{X}_T)-\mathcal{Y}_{T}^\pi \right\|^2$, we choose $0 < \epsilon < 1$ and $0< \delta < 1$ so that the corresponding coefficients are positive. Then, using \eqref{estimate:Deltaf}, \eqref{estimate:maxY}, and noticing that $(1-\epsilon^{-1} - \delta^{-1} ) <0$, we have
\begin{equation}
\begin{aligned}
& (1-\epsilon^{-1} - \delta^{-1})  \E \left\| \Delta \int_{t_0}^T f_s \mathrm{d}s \right\|^2  \\
\geq & (1-\epsilon^{-1} - \delta^{-1}) \left( Ch + T K_z^2 \E  \left[ \int_{t_0}^{T}  \|\Delta \mathcal{Z}_s\|^2 ds \right] +  2 T^2 K_y^2  \max_{0\leq i\leq N-1}  \E \|\Delta \mathcal{Y}_{t_i}\|^2 \right)   \\   
\geq & (1-\epsilon^{-1} - \delta^{-1})  \left( Ch + T K_z^2 \E  \left[ \int_{t_0}^{T} \|\Delta  \mathcal{Z}_s\|^2 ds \right] +  2 T^2 K_y^2  \left( Ch + \frac{3}{1-6T^2 K_y^2} \E \left\| \Delta \mathcal{Y}_{t_0} \right\|^2  \right. \right.\\
& \qquad \left. \left. + 
\frac{3\left(1+T K_z^2 \right)}{1-6 T^2 K_y^2} \E \left[ \int_0^T \|\Delta \mathcal{Z}_s\|^2  \mathrm{d}s  \right] \right)  \right) . 
\end{aligned}
\end{equation}
Substituting this lower bound into to \eqref{estimate:lossfun} yields,
\begin{equation}\label{estimate:final}
\begin{aligned}
\E \left\| g(\mathcal{X}_T) - \mathcal{Y}_{T}^\pi \right\|^2 
& \geq (1-\epsilon^{-1} - \delta^{-1}) Ch 
+ \left[ 1-\epsilon + (1-\epsilon^{-1} - \delta^{-1}) \frac{6 T^2 K_y^2 }{1-6T^2 K_y}  \right]   \E \left\| \Delta \mathcal{Y}_{t_0} \right\|^2  \\
& \quad + \left[ (1-\epsilon^{-1} - \delta^{-1}) \left( T K_z^2 + 2 T^2 K_y^2 \frac{3+3TK_z^2}{1-6T^2K_y^2} \right) + 1 -\delta \right]  \E \left[  \int_{t_0}^{T} \left\| \Delta \mathcal{Z}_s \right\|^2 \mathrm{d}s  \right] .
\end{aligned}
\end{equation}
Finally, to obtain the desired estimate, the coefficients in front of
\(\mathbb{E}\!\int_{t_0}^{T} \left\| \Delta \mathcal{Z}_s \right\|^2 \mathrm{d}s\)
and
\(\mathbb{E} \left\| \Delta \mathcal{Y}_{t_0} \right\|^2\)
in \eqref{estimate:final} must be positive. In addition to the conditions
\(0<\varepsilon<1\) and \(0<\delta<1\), this leads to the following conditions 
\begin{equation}\label{ineq:conditions}
(1-\epsilon^{-1} - \delta^{-1}) \left( T K_z^2 + 2 T^2 K_y^2 \frac{3+3TK_z^2}{1-6T^2K_y^2} \right) + 1 - \delta  >0 , \quad
1-\epsilon + (1-\epsilon^{-1} - \delta^{-1} )  \frac{6 T^2 K_y^2}{1-6T^2 K_y^2}   > 0 
\end{equation}
These conditions can, for instance, be satisfied by choosing
\(\varepsilon=\delta=\tfrac12\) and requiring
\[
T K_z^2
+
2 T^2 K_y^2 \frac{3+3T K_z^2}{1-6T^2 K_y^2}
< \frac{1}{6},
\]
which is clearly achievable when \(K_y\) and \(K_z\) are sufficiently small.
 
We move the negative term \((1-\varepsilon^{-1}-\delta^{-1})\,C h\) in \eqref{estimate:final} to the left-hand side, and invoke the \(L^2\)-regularity of the forward process \(\mathcal{X}\) to complete the proof. For the case \(\mathbb{E}\!\left\| g(\mathcal{X}_T^\pi) - \mathcal{Y}_T^\pi \right\|^2\), the result follows from the triangle inequality together with the Lipschitz continuity of \(g\).

\end{proof}

\begin{remark}
We note that Assumption~\ref{assume:Ky} can, in principle, be made fully explicit by deriving concrete upper bounds for the Lipschitz constants \(K_y\) and \(K_z\). However, obtaining optimal or near-optimal bounds in this way leads to a rather involved optimization problem.

Although this is not the main focus of the present paper, we briefly comment on how to set up this optimization problem. From the proof of Theorem~\ref{thm:deepBSDE}, the required conditions on \(K_y\) and \(K_z\) essentially arise from two sources: 
(i) the repeated application of Young’s inequality to control product terms, and 
(ii) the sign conditions imposed on certain coefficients to guarantee positivity or negativity where required. 

Therefore, one may introduce several auxiliary parameters \(\lambda_i\), in addition to $\epsilon$ and $\delta$, when applying Young’s inequality throughout the proof. This leads to a system of inequalities involving \(K_y\), \(K_z\), \(T\), \(\lambda_i\), \(\varepsilon\), and \(\delta\), which must be satisfied simultaneously, or a constrained optimisation problem if we aim to find the largest admissible values of \(K_y\) and \(K_z\) and $T$.
\end{remark}

With these results in place, we are now ready to give a simple proof of our method, which exploits structural properties of the deep BSDE method for a single BSDE, and the Lipschitz continuity of the compounding conditions $g_j$.

\begin{theorem}[Convergence of the Compound BSDE method]\label{main-thm}
Let Assumptions \ref{assum1}, \ref{assume: cadlag}, and \ref{assume:Ky} hold, and $(X, Y, Z)$ be the solution to \eqref{eq:compoundBSDE}. Then the approximated solution $(X^\pi, Y^\pi, Z^\pi)$ obtained by the Compound BSDE method \eqref{method:compoundBSDE} satisfies the following \emph{a posteriori} error estimate,
\begin{equation}\label{estimate:compoundBSDE}
\begin{aligned}
& \sup_{t \in [0, T]}  \mathbb{E}  \|X_{t} - X_{\pi(t)}^\pi \|^2 + 
\sum_{j=1}^M
\left( 
\sup_{t \in [T_{j-1}, T_j] }  \mathbb{E} \|Y_{j, t} - Y_{j, \pi(t)}^\pi \|^2 + 
\int_{T_{j-1}}^{T_j} \mathbb{E} \|Z_{j, t} - Z_{j, \pi(t)}^\pi \|^2  \mathrm{~d} t
\right) \\
&\hspace{4cm}  \leq C \left( h + \sum_{j=1}^{M-1} \mathbb{E} \|g_j( X_{T_{j}}^\pi, Y_{j+1, T_j}^\pi) - Y_{j, T_j}^\pi\|^2 + \mathbb{E} \|g_M(X_{T}^\pi) - Y_{M, T}^\pi\|^2  \right),
\end{aligned}
\end{equation}
where $C$ is a constant independent of the step size $h$.
\end{theorem}

\begin{proof}

We apply Theorem \ref{estimate:deepBSDE} to the $j$-th BSDE in the Compound BSDE formulation \eqref{eq:compoundBSDE}, for $j=1, 2, \ldots, M-1$,
\begin{equation}\label{ineq:main}
\begin{aligned}
& \sup_{t \in [T_{j-1}, T_j] }  \left( \mathbb{E} \|X_{j, t} - X_{j, \pi(t)}^\pi \|^2  +  \mathbb{E} \|Y_{j, t} - Y_{j, \pi(t)}^\pi \|^2  \right) + \int_{T_{j-1}}^{T_j} \mathbb{E} \|Z_{j, t} - Z_{j, \pi(t)}^\pi \|^2  \mathrm{~d} t   \\
\leq & C \mathbb{E} \|Y_{j, T_j}^{\pi} - g_j(X_{T_j}, Y_{j+1, T_j}) \|^2  + Ch  \\
\leq & C \mathbb{E} \|Y_{j, T_j }^{\pi} - g_j(X_{T_j}^\pi, Y_{j+1, T_j }^{\pi}) \|^2 + C \mathbb{E} \| g_j( X_{T_j}^\pi, Y_{j+1, T_j}^{\pi} )  - g_j( X_{T_j}, Y_{j+1, T_j} ) \|^2  + Ch \\
\leq & C \mathbb{E} \|Y_{j, T_j }^{\pi} - g_j( X_{T_j}^\pi, Y_{j+1, T_j}^{\pi} ) \|^2  + C \mathbb{E} \| Y_{j+1, T_j}^{\pi} - Y_{j+1,T_j} \|^2 + C \mathbb{E} \| X_{T_j}^\pi -  X_{T_j} \|^2 + Ch  \\
\leq & C \mathbb{E} \|Y_{j, T_j }^{\pi} - g_j( X_{T_j}^\pi, Y_{j+1, T_j}^{\pi} ) \|^2 + C \mathbb{E} \| Y_{j+1, T_j}^{\pi} - Y_{j+1,T_j} \|^2  + Ch ,
\end{aligned}
\end{equation}
where $C$ denote a generic constant that may change from line to line, and we use the triangle inequality in the second inequality to split the term, since we do not have access to $Y_{j+1, T_j}$, which is part of the solution of \eqref{eq:compoundBSDE} that needs to be solved for. The remaining derivations follow from the Lipschitz continuity of $g_j$ and the $L^2$-regularity of the $X$ process.  

For the case $j=M$, note that typically we have access to the samples for $g_M(X_T^\pi)$ via numerical methods for the SDE of $X$, and thus applying Theorem \ref{estimate:deepBSDE} again gives,
\begin{equation}\label{ineq:main2}
\begin{aligned}
& \sup_{t \in [T_{M-1}, T_M] }  \left( \mathbb{E} \|X_{M, t} - X_{M, \pi(t)}^\pi \|^2 + \mathbb{E}  \|Y_{M, t} - Y_{M, \pi(t)}^\pi \|^2  \right) + \int_{T_{M-1}}^{T_M} \mathbb{E}  \|Z_{M, t} - Z_{M, \pi(t)}^\pi  \|^2  \mathrm{~d} t  \\
\leq & C \left( \mathbb{E} \|Y_{M, T_M }^{\pi} - g_j(X_{T_M}^\pi)  \|^2 + h \right) .
\end{aligned}
\end{equation}

Fixing some $1\leq j\leq M-2$, we notice that the second term in the last inequality of \eqref{ineq:main} can be controlled by using the estimate \eqref{ineq:main} for the $(j+1)$-th BSDE, as this second term in the $j$-th case will appear on the left-hand side in the $j+1$-case. Through a similar recursive argument, all these second terms in the last lines of \eqref{ineq:main}, for $j=1,2\cdots, M-1$, will be bounded by the right-hand-side of \eqref{ineq:main2}, provided that all the first terms $\mathbb{E} \|Y_{j, T_j }^{\pi} - g_j( X_{T_j}^\pi, Y_{j+1, T_j}^{\pi} ) \|$ in \eqref{ineq:main} can be controlled. Therefore, we only need to keep all these first terms in \eqref{ineq:main}  and the right-hand side in \eqref{ineq:main2}, and obtain the desired estimate. 

\end{proof}

\begin{remark}
Inspecting the proof, we observe that it is possible to include the $Z$ process in $g_j$ for $1 \le j \le M-1$, in the compound BSDE formulation \eqref{eq:compoundBSDE}. This extension, however, would require stronger regularity assumptions on the decoupling field associated with the $Z$ process; for example, one may assume that the function $(\nabla_x u_j(t, X_t))^{\mathrm{T}} \sigma(t, X_t)$ is uniformly Lipschitz continuous in $X_t$ for all $j$. A rigorous treatment of this extension is left for future research.
\end{remark}

\section{Applications in derivatives pricing}\label{sec4}

In this section, we illustrate how the compound BSDE formulation~\eqref{eq:compoundBSDE} and the associated \emph{fully forward} method~\eqref{method:compoundBSDE} can be used for option pricing and optimal stopping problems. The key ingredient is the flexibility in the choice of the compounding functions $g_j$. Unlike standard BSDE formulations, where only the terminal condition at $T$ is prescribed, our framework allows one to impose conditions at intermediate times $T_j$ as well. This makes it possible to represent a wide range of payoff structures and early-exercise features within a single template.

We next discuss several representative choices of $g_j$ and explain the corresponding pricing or stopping problems.

\case{$g_j(x,y)=y$ for all $j$.}\label{case1}

With $g_j(x,y)=y$, the compound BSDE~\eqref{eq:compoundBSDE} collapses to a standard BSDE of the form~\eqref{eq:BSDE}, and the method~\eqref{method:compoundBSDE} reduces to the classical deep BSDE method~\eqref{eq:deepBSDE}. Indeed, the intermediate loss terms
$$ 
\E\bigl \|Y_{j,T_j}^\pi-g_j( Y_{j+1,T_j}^\pi)\bigr\|^2
=\E\bigl\|Y_{j,T_j}^\pi-Y_{j+1,T_j}^\pi\bigr\|^2
$$
enforce consistency across the sub-intervals; when these terms are small, we have
$Y_{j,T_j}^\pi\approx Y_{j+1,T_j}^\pi$ for $j=1,\dots,M-1$.
In that regime, the dominant contribution to the loss is the terminal term
$$
\E\bigl\|Y_{M,T_M}^\pi-g_M(X_{T_M}^\pi)\bigr\|^2,
$$
which is exactly the deep BSDE objective.

As a canonical example, consider the pricing of a European put option in the Black--Scholes--Merton model under the risk-neutral measure. Let $X$ follow a geometric Brownian motion and choose
$f_j(t,x,y,z)=- ry$ for all $j$, where $r$ is the risk-free rate. With terminal condition $g_M(x)=(K-x)^+$, the BSDE solution satisfies $Y_t=V(t,X_t)$ and
$Z_t=V_x(t,X_t)^\top \sigma(t,X_t)$ by the nonlinear Feynman--Kac representation \eqref{eq:feynman-kac}, so that the processes $(Y,Z)$ correspond to the option price and the scaled option delta over the whole time horizon $[0, T]$, respectively.

\case{$g_j(x,y)=\max(y-K_j,0)$ for constants $K_j$ and $1\leq j\leq M-1$.}\label{case2}

For $g_j(x,y)=(y-K_j)^+ \equiv \max(y-K_j, 0)$, the compound BSDE \eqref{eq:compoundBSDE} can be interpreted as a discrete-time reflection mechanism at the prescribed times $T_j$, with constant thresholds $K_j$. In particular, for $j=1,\dots,M-1$, the compounding condition
\[
Y_{j,T_j}=(Y_{j+1,T_j}-K_j)^+
\]
maps the ''next-stage value'' $Y_{j+1,T_j}$ into the previous stage through a payoff-type nonlinearity.

A natural application is the pricing of compound options. For illustration, consider a call-on-call option with $M=2$. Let $V_{\mathrm{coc}}(t, x_t)$ denote the call-on-call option price at $(t, x_t)$ with maturity $T_1$, and let $V_{\mathrm{c}}(t, x_t)$ denote the vanilla European call option price at $(t, x_t)$ with maturity $T_2>T_1$. At time $T_1$, the call-on-call option holder may buy the vanilla call by paying the strike $K_1$, hence the holder's payoff at $T_1$ is
\[
V_{\mathrm{coc}}(T_1, x) = (V_{\mathrm{c}}(T_1, x)-K_1)^+.
\]
If exercised, the remaining position is the vanilla call with a terminal payoff
$V_{\mathrm{c}}(T_2, x)=(X_{T_2}-K_2)^+$.

Compound options also arise naturally as real options in corporate finance, energy investment, and R\&D investment. In such settings, a call-on-call represents the right to acquire an investment opportunity at a future date rather than committing immediately. Paying an initial premium preserves the flexibility to invest only if conditions evolve favourably, thereby limiting downside risk while retaining upside potential. This captures the value of waiting, learning, and managerial flexibility under uncertainty.

To apply~\eqref{method:compoundBSDE} in the Black--Scholes--Merton framework, we take the same geometric Brownian motion $X$ and the driver $f_j(t,x,y,z)=-ry$ as in Case~\ref{case1}. Then, by the Feynman--Kac representation, $(Y_1,Y_2)$ corresponds to $(V_{\mathrm{coc}},V_{\mathrm{c}})$ in the call-on-call example. The same construction extends directly to other compound options and to $M$-fold compound options with $M\ge2$; for instance, a put-on-put can be handled by replacing the compounding payoff with $g_j(x,y)=(K_j-y)^+$ while keeping the dynamics and driver unchanged.

\case{$g_j(x,y) = \max\bigl(y,(x-K_j)^+\bigr)$ for constants $K_j$ and $1\leq j\leq M-1$.} \label{case3}

This choice generalises Case~\ref{case2} by allowing the compounding condition to incorporate an \emph{immediate exercise value} $(x-K_j)^+$ that depends on the state $X_{T_j}$. The resulting condition compares continuation and exercise values at each $T_j$, which is precisely the structure of Bermudan-style optimal stopping.

To see this, consider a Bermudan call option with exercise dates
$\mathcal{T}\coloneqq\{T_j:\,j=1,\dots,M\}$.
At each date $T_j$, the option value satisfies the dynamic programming principle
\begin{equation}\label{case3_cond}
Y_{j,T_j}=\max\!\Bigl(Y_{j+1,T_j},\, (X_{T_j}-K_j)^+\Bigr),
\end{equation}
where $Y_{j+1,T_j}$ represents the continuation value. Therefore, if the loss in~\eqref{method:compoundBSDE} is sufficiently small, the compounding conditions enforce~\eqref{case3_cond} to high accuracy, and the resulting learned solution provides the Bermudan price. An American option can be approached by refining the grid of exercise times, so that the Bermudan approximation becomes accurate.

\begin{remark}
In light of the above discussion, it is natural to relate our compound BSDE method \eqref{method:compoundBSDE} to the so-called discretely reflected BSDEs \cite{el1997backward, zhang2017backward}. To illustrate this connection, consider the Bermudan option pricing problem described in Case \ref{case3} and use the same set $\mathcal{T}$. This problem can be modelled by the following discretely reflected BSDE,
\begin{equation}\label{eq:reflectedBSDE}
\left\{
\begin{aligned}
Y_{T_M} & = g_M(X_{T_M}), \\
\widetilde{Y}_t
&= Y_{T_j}
+ \int_t^{T_j} r\, Y_s \, ds
- \int_t^{T_j} Z_s \, dW_s,
\qquad t \in [T_{j-1}, T_j), \\
Y_{T_{j-1}}
&= \max \!\Big( \widetilde{Y}_{T_{j-1}}, (X_{T_{j-1}} - K_j )^+ \Big)
\;\coloneqq\;
\mathcal{R}_{j-1}\!\big(
X_{T_{j-1}}, \widetilde{Y}_{T_{j-1}}
\big).
\end{aligned}
\right.
\end{equation}
Such equations are typically solved in a backward manner; see the references discussed in Section \ref{sec:intro}.

Comparing \eqref{eq:reflectedBSDE} with the compound BSDE formulation \eqref{eq:compoundBSDE}, we observe that the reflection operator $\mathcal{R}_{j-1}$ plays essentially the same role as the function $g_{j-1}$ in our framework. From this viewpoint, our method can also be interpreted as a numerical scheme for discretely reflected BSDEs of the form \eqref{eq:reflectedBSDE} with Lipschitz $\mathcal{R}_{j-1}$. Moreover, the reason why our algorithm can be implemented in a forward training manner becomes clear from this perspective: instead of enforcing the hard constraint via directly computing  
\[
Y_{T_{j-1}} = \mathcal{R}_{j-1}\!\big(
X_{T_{j-1}}, \widetilde{Y}_{T_{j-1}}
\big),
\]
as done, for example, by the max method in \cite{gobet2008numerical}, or methods in \cite{gao2023convergence} \cite{hure2020deep}, we adopt a relaxed formulation by incorporating this reflection relation into the objective functional.
\end{remark}

\section{Numerical examples}\label{sec5}

In this section, we examine the accuracy and efficiency of the proposed Compound BSDE method for several option pricing problems.  In particular, we include Bermudan-style options in high dimensions.

To assess accuracy, we consider examples under geometric Brownian motion (GBM) dynamics. In this setting, compound options admit analytical formulas, and reliable benchmark methods are available for Bermudan options. Let $X_t=(X_t^1,\dots,X_t^{d_1})^\top$ denote the underlying asset-price vector. Under the risk-neutral measure, we assume that $X$ follows
\begin{equation}\label{sde-gbm}
d X_t = \operatorname{diag}(X_t)\bigl((r\mathbf{1}_{d_1}-q)\,dt+\Sigma\,dW_t\bigr),
\qquad X_0=x_0,
\end{equation}
where $\mathbf{1}_{d_1}\in\R^{d_1}$ is the vector of ones, $r$ is the risk-free rate, $q=(q_1,\dots,q_{d_1})^\top$ is the vector of continuous dividend yields, $\Sigma\in\R^{d_1\times d_1}$ is the volatility matrix, and $W_t$ is a $d_1$-dimensional standard Brownian motion with independent components and thus $d_3=d_1$. The dimension $d_1$ and the parameter choices are specified in each numerical experiment below.

We also investigate the accuracy of the bound in~\eqref{estimate:compoundBSDE}. When a reference (analytical or highly accurate) solution $(X,Y,Z)$ to \eqref{eq:compoundBSDE} is available, we report the following error metrics for the numerical solution $(X^\pi, Y^\pi, Z^\pi)$ generated by the Compound BSDE method \eqref{method:compoundBSDE}, 
\begin{equation}\label{error-metric}
\begin{aligned}
& \operatorname{Err}(X) \coloneqq \sup_{t \in [0, T]}  \E \left\|X_{\pi(t)} - X_{\pi(t)}^\pi \right\|^2, \qquad
\operatorname{Err}(Y) \coloneqq  \sup_{1\leq j\leq M} \sup_{t \in [T_{j-1}, T_j]}  \E \left\|Y_{j, {\pi(t)}} - Y_{j, {\pi(t)}}^\pi \right\|^2, \\
& \operatorname{Err}(Z) \coloneqq  \sum_{j=1}^M  \sum_{t_i\in [T_{j-1}, T_j]}  \E \left\|Z_{j, t_i} - Z_{j,t_i}^\pi \right\|^2 h , \qquad
\operatorname{Total\ Err} \coloneqq \operatorname{Err}(X)+\operatorname{Err}(Y)+\operatorname{Err}(Z).
\end{aligned}
\end{equation}
where the expectation is approximated by the sample mean via Monte Carlo. Since our primary focus is option pricing, we additionally report the mean squared error of the option price and delta at time $t=0$, together with the corresponding relative mean squared errors obtained by normalising with the mean squared reference quantities.

Regarding implementation, we use a fully connected neural network with two hidden layers, each with $10+d_1$ neurons, and $\tanh$ activations. The batch size and validation size are both $5000$, corresponding to the number of simulated sample paths used per parameter update. We train with the Adam optimizer and an exponential learning-rate decay schedule with an initial learning rate $0.01$. All computations are performed in PyTorch~2.4 on a MacBook Pro with an Apple M1 chip and 16~GB RAM. The code for the numerical experiments can be found in the GitHub repository: \href{https://github.com/zhipeng-huang/thecompoundbsde}{https://github.com/zhipeng-huang/thecompoundbsde}.

\subsection{Plain compound options with $M=2$}\label{subsecex1}

We begin with a simple one-dimensional setting, namely plain compound options with $M=2$: call-on-call, call-on-put, put-on-call, and put-on-put. Under geometric Brownian motion dynamics and the risk-neutral measure, all four admit analytical formulas; see \cite{geske1979valuation} for the call-on-call case. Let $V_{\text{coc}}(t,x_t)$ denote the price of the call-on-call option at current state $(t,x_t)$ for $0\le t<T_1$, with strike prices $K_j$ and exercise times $T_j$ for $j=1,2$, and we use analogous notation for the other three contracts.

The reference option prices are given by, 
\begin{equation}\label{plaincompound_formula}
\begin{aligned}
V_{\text{coc}} (t, x_t) & =  x_t e^{-q (T_2 - t)} \Phi_2(a_1,b_1; P^{(2)}) - K_2 e^{-r (T_2 -t)} \Phi_2(a_2,b_2;P^{(2)}) - K_1 e^{-r (T_1 - t)} \Phi(a_2),  \\
V_{\text{cop}} (t, x_t) & =  K_2 e^{-r (T_2 - t)} \Phi_2(a_2,-b_2;-P^{(2)}) - x_t e^{-q (T_2 - t)} \Phi_2(a_1,-b_1;-P^{(2)}) - K_1 e^{-r (T_1 - t)} \Phi(a_2), \\
V_{\text{poc}} (t, x_t) & =  K_1 e^{-r (T_1 - t)} \Phi(-a_2) - x_t e^{-q (T_2 - t)} \Phi_2(-a_1,b_1;-P^{(2)}) + K_2 e^{-r (T_2 - t)} \Phi_2(-a_2,b_2;-P^{(2)}), \\
V_{\text{pop}} (t, x_t) & =  K_1 e^{-r (T_1 - t)} \Phi(-a_2) + x_t e^{-q (T_2 - t)} \Phi_2(-a_1,-b_1;P^{(2)}) - K_2 e^{-r (T_2 - t)} \Phi_2(-a_2,-b_2;P^{(2)}),
\end{aligned}
\end{equation}
and we shall obtain the corresponding option deltas by directly computing the derivatives with respect to $x_t$, respectively. Here, $\Phi_2(\cdot, \cdot; P^{(2)})$ is the bivariate standard normal CDF with correlation matrix $P^{(2)}$, and $\Phi$ is the univariate standard normal CDF. The remaining quantities are
\begin{equation}
\begin{aligned}
& a_1 = \frac{\ln ( x_t / K_1^* ) + \left(r - q +\frac{1}{2} \sigma^2\right) (T_1 - t)}{\sigma \sqrt{T_1-t}},
\qquad
a_2 = a_1- \sigma \sqrt{T_1 - t}, \\
& b_1 = \frac{\ln \left( x_t /K_2 \right)+ \left(r - q +\frac{1}{2} \sigma^2\right) (T_2 - t)}{\sigma \sqrt{T_2-t}},
\qquad
b_2 = b_1 - \sigma \sqrt{T_2 - t},
\qquad
P^{(2)}_{12} = P^{(2)}_{21} = \sqrt{\frac{T_1-t}{T_2-t}} .
\end{aligned}
\end{equation}
The constant $K_1^*$ is the critical underlying level at time $T_1$ for which the inner option value equals the strike price $K_1$ of the outer option. More precisely, $K_1^*$ solves
\begin{equation}\label{eq-threshold}
\begin{aligned}
& V_{\text{c}}(T_1, K_1^* ; T_2, K_2 ) = K_1,
\qquad \text{for call-on-call or put-on-call},\\
& V_{\text{p}}(T_1, K_1^* ; T_2, K_2 ) = K_1,
\qquad \text{for call-on-put or put-on-put},
\end{aligned}
\end{equation}
where $V_{\text{c}}(t, x_t; T_2, K_2 )$ and $V_{\text{p}}(t, x_t; T_2, K_2 )$ denote the Black--Scholes call and put prices at time $t$ with underlying level $x$, strike $K_2$ and expires at $T_2$. The root-finding problems in~\eqref{eq-threshold} are solved efficiently via Brent's method implemented in SciPy.

To implement~\eqref{method:compoundBSDE}, we follow Case~\ref{case2} in Section~\ref{sec4}. Under the risk-neutral measure, we set
$f_1(t,x,y,z)=f_2(t,x,y,z)=-ry$
and choose $g_j$ according to the payoff type. For example, for a \emph{call-on-call} option we take
\[
g_1(x,y)=(y-K_1)^+,
\qquad
g_2(x)=(x-K_2)^+.
\]
Analogous choices apply to call-on-put, put-on-call, and put-on-put by changing the outer and inner payoffs.

In our experiments, we use the following parameters and test our method over different $N \in \{10,20,30,40,50\}$,
\[
r = 0.03,\quad q = 0,\quad \Sigma = 0.2,\quad x_0 = 14,\quad K_1 = 1,\quad K_2 = 14,\quad T_1=0.2,\quad T_2=0.4.
\]

Table~\ref{tab:plaincompound} presents the results for all four plain compound options with $N=50$. We compare the estimated price and delta at $t=0$ for $j=1$, denoted by $Y_{1,0}^\pi$ and $\Delta_{1,0}^\pi$, with the reference values $Y_{1,0}$ and $\Delta_{1,0}$ obtained from~\eqref{plaincompound_formula} and their analytical derivatives. We recall that the BSDE solution $(Y,Z)$ coincides with the option price and the appropriately scaled option delta via the Feynman--Kac formula, and we use this relation to recover the delta from $Z^\pi$. The corresponding relative mean squared errors are also reported. Overall, the method produces highly accurate estimates for both price and delta across all four contracts, with relative errors ranging from $1.04\times 10^{-3}$ down to $3.82\times 10^{-7}$.

To assess the bound in~\eqref{estimate:compoundBSDE}, we also compute the pathwise error metrics in~\eqref{error-metric} for different $N$. Figure~\ref{fig:plain-compound} shows the results for two representative contracts: call-on-call and put-on-put. All error metrics decrease approximately at first order in the step size $h=T/N$: increasing $N$ from $10$ to $50$ reduces the errors by roughly a factor $5$. This behaviour is consistent with the convergence estimate for~\eqref{method:compoundBSDE}. Moreover, the empirical total error (red curve) follows the same trend as the theoretical bound (purple curve), supporting the sharpness of~\eqref{estimate:compoundBSDE} in this setting.

\begin{table}[htbp]
\centering
\caption{Prices and deltas of four types of plain compound options at $t=0$.}
\label{tab:plaincompound}
\begin{tabular}{ccccccc}
\toprule
Option Types & $Y_{1,0}^\pi$ & $Y_{1,0}$ & reMSE & $\Delta_{1,0}^\pi$ & $\Delta_{1,0}$ & reMSE \\
\midrule
call-on-call & 0.227 & 0.224 & 1.491e-04 & 0.289 & 0.291 & 5.621e-05 \\
call-on-put  & 0.118 & 0.120 & 2.209e-04 & -0.163 & -0.168 & 1.044e-03 \\
put-on-call  & 0.432 & 0.430 & 3.029e-05 & -0.270 & -0.272 & 2.987e-05 \\
put-on-put   & 0.492 & 0.492 & 3.822e-07 & 0.267 & 0.269 & 5.460e-05 \\
\bottomrule
\end{tabular}

\end{table}

\begin{figure}[htbp]
    \centering
    \begin{subfigure}{0.45\textwidth}
        \centering
        \includegraphics[width=\linewidth]{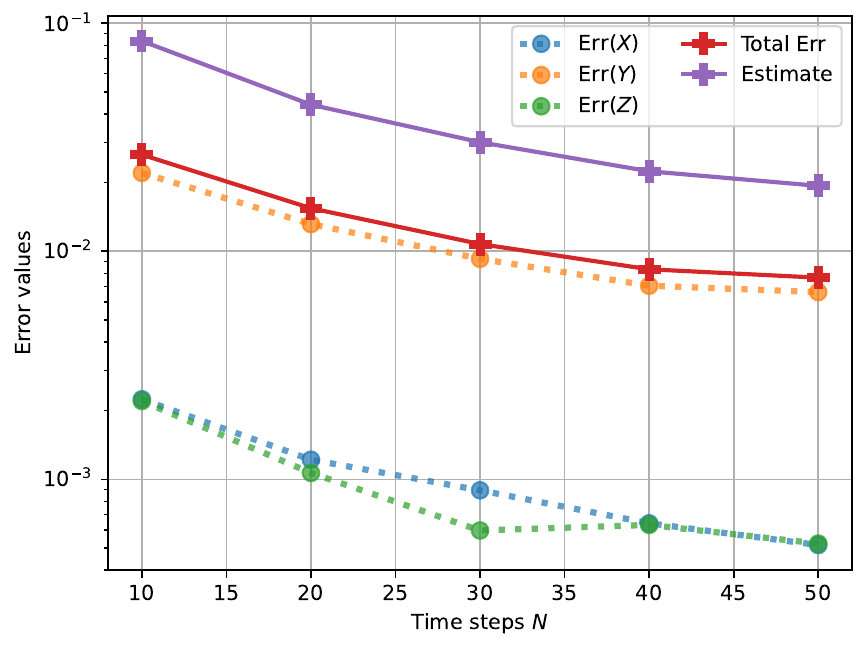}
        \caption{call-on-call}
    \end{subfigure}
    \hfill
    \begin{subfigure}{0.45\textwidth}
        \centering
        \includegraphics[width=\linewidth]{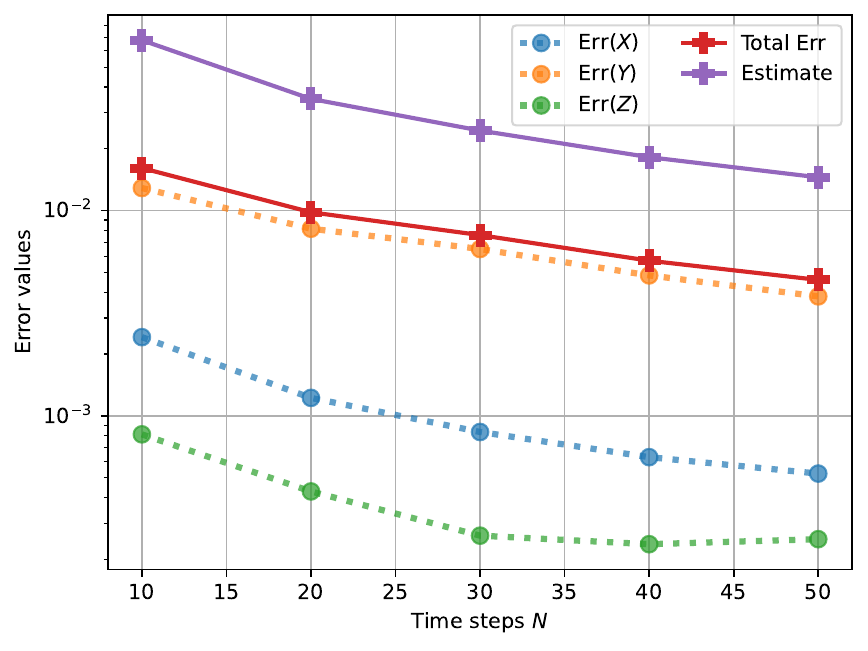}
        \caption{put-on-put}
    \end{subfigure}
    \caption{Convergence in $N$ for plain compound options.}
    \label{fig:plain-compound}
\end{figure}
\subsection{$M$-fold compound options}\label{subsecex2}

We next consider the $M$-fold compound call option for $M\ge2$ in the one-dimensional setting. Compared with the plain case with $M=2$, the $M$-fold structure allows for multiple sequential exercise decisions and is therefore closer to many real-world investment settings, where decisions are made gradually rather than in a single step.

Let $V_{\text{c}}^{(M)}(t,x_t)$ denote the price of the $M$-fold compound call at current state $(t,x_t)$ for $0\le t<T_1$. As reference values, we use the analytical expressions in \cite{thomassen2002sensitivity,cassimon2004valuation}, 
\begin{equation}
V_{\text{c}}^{(M)}(t,x_t)
= x_t e^{-q(T_M-t)}\, \Phi_M\!\left(a_1,\dots,a_M; P^{(M)}\right)
- \sum_{j=1}^M K_j e^{-r(T_j-t)}\, \Phi_j\!\left(b_1,\dots,b_j; P^{(j)}\right),
\end{equation}
and we obtain the delta by computing the derivative. Here $\Phi_j(\cdot;P^{(j)})$ denotes the $j$-dimensional standard normal CDF with correlation matrix $P^{(j)}$. The parameters are given, for $j=1,\dots,M$, by
\begin{equation}
\left(P^{(j)}\right)_{i\ell}=\sqrt{\frac{T_i-t}{T_\ell-t}},
\qquad i<\ell,
\end{equation}
\begin{equation}
b_j = \frac{\ln ( x_t/ K_j^* ) + \left(r - q - \frac{1}{2} \sigma^2\right) (T_j-t)}{\sigma \sqrt{T_j-t}},
\qquad
a_j = b_j + \sigma \sqrt{T_j-t}.
\end{equation}
To determine the critical levels, we set $K_M^*=K_M$, and determine $K_j^*$ recursively by solving
\begin{equation}\label{m-fold:criticalvalues}
V_{\text{c}}^{(M-j)}\!\left(T_j, K_j^*\right) =K_j, \qquad j=1,2,\dots,M-1,
\end{equation}
where $V_{\text{c}}^{(M-j)}\!\left(T_j, K_j^*\right)$ is the $(M-j)$-fold compound option price at time $T_j$, with an appropriate sequence of exercise times and strike prices, and in the case $M-j = 1$ it reduces to the Black-Scholes call price. We remark that although \eqref{m-fold:criticalvalues} can be solved numerically similarly to Subsection \ref{subsecex1}, the computational cost is high if one aims to compute the option prices over the whole time horizon $[0, T]$ when $M$ is large.


We solve the $M$-fold call problem for $M\in\{2,3,4,5\}$ using
\[
r = 0.03,\quad q = 0,\quad \Sigma = 0.2,\quad x_0 = 5,\quad h = 0.05, \quad K_j = 1,\quad T_j=j,\quad 1\le j\le M,
\]
and we keep the step size $h$ fixed across different $M$ by increasing the total number of time steps $N$ accordingly.

Table~\ref{tab:mfoldcompound} summarises the results. The accuracy is comparable to that in Subsection~\ref{subsecex1} for both price and delta. Moreover, at fixed $h$, performance remains essentially stable as $M$ increases, even though the loss contains more compounding terms. This suggests that, in this regime, the number of compounding conditions is not the main limitation of the method. 

\begin{table}[htbp]
\centering
\caption{Prices and deltas of $M$-fold compound options with different $M$ at $t=0$.}
\label{tab:mfoldcompound}
\begin{tabular}{ccccccc}
\toprule
$M$ & $Y_{1,0}^\pi$ & $Y_{1,0}$ &  reMSE & $\Delta_{1,0}^\pi$ & $\Delta_{1,0}$ & reMSE \\
\midrule
2 & 3.077 & 3.088 & 1.153e-05 & 0.994 & 1.000 & 3.037e-05 \\
3 & 2.173 & 2.174 & 2.337e-07 & 0.988 & 0.998 & 1.115e-04 \\
4 & 1.308 & 1.315 & 2.212e-05 & 0.933 & 0.942 & 9.379e-05 \\
5 & 0.639 & 0.640 & 7.892e-06 & 0.693 & 0.707 & 4.153e-04 \\
\bottomrule
\end{tabular}

\end{table}

\subsection{Bermudan geometric basket put option}\label{subsecex3}

We finally consider a Bermudan-style geometric basket put option. This contract is a standard benchmark for high-dimensional Bermudan pricing, since it can be reduced to an equivalent one-dimensional problem, providing an accurate reference value at modest computational cost.

We compute the reference price and delta by applying a binomial tree with a sufficiently large number of time steps to the equivalent one-dimensional formulation. Under the $d_1$-dimensional GBM dynamics~\eqref{sde-gbm}, the geometric average
\[
\hat{X}_t \coloneqq \left(\prod_{i=1}^{d_1} X_t^i\right)^{1/d_1}
\]
is itself a GBM driven by a one-dimensional Brownian motion $\{B_t\}_{0\le t\le T}$ with parameters
\begin{equation}
\hat{x}_0 = \left(\prod_{i=1}^{d_1} x_0^i\right)^{1/d_1},\quad
\hat{r}=r,\quad
\hat{q} = \frac{1}{d_1}\sum_{i=1}^{d_1}\!\left(q_i+\frac{1}{2}\sigma_i^2\right)-\frac{1}{2}\hat{\sigma}^2,\quad
\hat{\sigma}^2 = \frac{1}{d_1^2}\!\left( \sum_{i,j=1}^{d_1}\sigma_i\sigma_j\rho_{ij}\right),
\end{equation}
and payoff $(K-\hat{X}_t)^+$ at exercise times. Here $\rho_{ij}$ denotes the correlation coefficient between the $i$-th and $j$-th components of the original $d_1$-dimensional Brownian motion $W$.  

To run our compound BSDE method under the risk-neutral measure, we again take $f_j(t,x,y,z)=-ry$ and specify the compounding conditions to encode Bermudan exercise:
\[
g_j(x,y)=\max\!\left(y,\ \left(K_j-\left(\prod_{i=1}^{d_1} x^i\right)^{1/d_1}\right)^+\right),
\qquad
g_M(x)=\left(K_M-\left(\prod_{i=1}^{d_1} x^i\right)^{1/d_1}\right)^+.
\]
We test the method across different dimensions $d_1$ and different numbers of time steps $N$. Let $\mathcal{T}=\{0.1,0.2,0.3,0.4,0.5\}$ be the set of early exercise dates of the option, and we set $\rho_{ij} = 0$ for $i\neq j$ and use the parameters for the test,
\begin{equation}
r = 0.02,\quad q=0,\quad \Sigma = \operatorname{diag}(0.2\cdot \mathbf{1}_{d_1} ),\quad X_0 = 49 \cdot \mathbf{1}_{d_1},\quad K_j = 50,\quad T_j = 0.1 \cdot j ,\quad
M = 5
\end{equation}
Table~\ref{tab:bermudan1} reports prices and deltas at $t=0$ for varying $d_1$ with fixed $N=100$, while Table~\ref{tab:bermudan2} reports results for varying $N$ with fixed $d_1=20$. Due to space constraints, we report the range of the $d_1$ components of $\Delta_{1, 0}^\pi$. Table~\ref{tab:bermudan1} indicates that the method achieves good accuracy for both price and delta up to dimension $d_1=50$, with relative errors remaining stable as the dimension increases. In Table~\ref{tab:bermudan2}, the relative errors for both price and delta decrease as $N$ increases, which is consistent with the convergence behaviour predicted by~\eqref{estimate:compoundBSDE}.

\begin{table}[htbp]
\centering
\caption{Results for the Bermudan geometric basket put option over different $d_1$ at $t=0$.}
\label{tab:bermudan1}
\begin{tabular}{ccccccc}
\toprule
$d_1$ & $Y_{1,0}^\pi$ & $Y_{1,0}$ &  reMSE & $\Delta_{1,0}^\pi$ & $\Delta_{1,0}$ & reMSE \\
\midrule
1  & 3.101 & 3.071 & 9.699e-05 & [-0.513, -0.513] & $-0.510 \cdot \mathbf{1}_{d_1}$ & 3.178e-05 \\
5  & 1.762 & 1.745 & 9.233e-05 & [-0.122, -0.121] & $-0.121 \cdot \mathbf{1}_{d_1}$ & 5.673e-05 \\
10 & 1.446 & 1.432 & 9.943e-05 & [-0.067, -0.066] & $-0.066 \cdot \mathbf{1}_{d_1}$ & 7.713e-05 \\
20 & 1.236 & 1.223 & 1.077e-04 & [-0.037, -0.036] & $-0.036 \cdot \mathbf{1}_{d_1}$ & 2.321e-05 \\
30 & 1.153 & 1.140 & 1.295e-04 & [-0.026, -0.026] & $-0.026 \cdot \mathbf{1}_{d_1}$ & 3.413e-05 \\
40 & 1.105 & 1.095 & 8.627e-05 & [-0.020, -0.020] & $-0.020 \cdot \mathbf{1}_{d_1}$ & 1.479e-05 \\
50 & 1.076 & 1.067 & 6.196e-05 & [-0.017, -0.017] & $-0.017 \cdot \mathbf{1}_{d_1}$ & 3.361e-05 \\
\bottomrule
\end{tabular}

\end{table}

\begin{table}[htbp]
\centering
\caption{Results for the Bermudan geometric basket put option over different $N$ at $t=0$.}
\label{tab:bermudan2}
\begin{tabular}{ccccccc}
\toprule
$N$ & $Y_{1,0}^\pi$ & $Y_{1,0}$ &  reMSE & $\Delta_{1,0}^\pi$ & $\Delta_{1,0}$ & reMSE \\
\midrule
10  & 1.264 & 1.223 & 1.087e-03 & [-0.037, -0.037] & $-0.036 \cdot \mathbf{1}_{20}$ & 2.424e-04 \\
20  & 1.240 & 1.223 & 1.947e-04 & [-0.037, -0.037] & $-0.036 \cdot \mathbf{1}_{20}$ & 1.283e-04 \\
30  & 1.244 & 1.223 & 2.934e-04 & [-0.037, -0.037] & $-0.036 \cdot \mathbf{1}_{20}$ & 7.129e-05 \\
40  & 1.240 & 1.223 & 1.908e-04 & [-0.037, -0.036] & $-0.036 \cdot \mathbf{1}_{20}$ & 5.242e-05 \\
50  & 1.237 & 1.223 & 1.330e-04 & [-0.037, -0.037] & $-0.036 \cdot \mathbf{1}_{20}$ & 7.086e-05 \\
60  & 1.238 & 1.223 & 1.356e-04 & [-0.037, -0.037] & $-0.036 \cdot \mathbf{1}_{20}$ & 5.234e-05 \\
70  & 1.236 & 1.223 & 1.168e-04 & [-0.037, -0.036] & $-0.036 \cdot \mathbf{1}_{20}$ & 3.327e-05 \\
80  & 1.236 & 1.223 & 1.063e-04 & [-0.037, -0.037] & $-0.036 \cdot \mathbf{1}_{20}$ & 9.738e-05 \\
90  & 1.234 & 1.223 & 7.871e-05 & [-0.037, -0.036] & $-0.036 \cdot \mathbf{1}_{20}$ & 3.880e-05 \\
100 & 1.233 & 1.223 & 6.568e-05 & [-0.037, -0.036] & $-0.036 \cdot \mathbf{1}_{20}$ & 2.171e-05 \\
\bottomrule
\end{tabular}

\end{table}

\section{Conclusion}\label{sec6}

This paper introduced the compound BSDE formulation and the associated Compound BSDE method, a fully forward deep-learning approach for solving a broad class of pricing and optimal stopping problems in finance. The key idea is to represent multi-stage payoffs and exercise features through compounding conditions at compounding times, so that several BSDEs are coupled or compounded through these conditions and can be learned simultaneously. In contrast to many backward-style schemes for early-exercise products, our method keeps the simulation and the learning direction fully forward and enforces all compounding and terminal conditions through a single objective functional. Under standard Lipschitz and time-regularity assumptions, we established well-posedness of the compound BSDE, together with $L^2$-regularity properties that are needed for numerical analysis. Building on an \emph{a posteriori} error estimate for the deep BSDE method for a single BSDE, we derived a corresponding \emph{a posteriori} estimate for the compound method, showing that the discretisation and approximation errors are controlled by the step size and by the training loss. The numerical experiments confirm our theoretical results and demonstrate that the method is capable of handling problems with multiple folds as well as early-exercise features in high-dimensional settings. These results suggest that the compound BSDE perspective provides a unified numerical framework for both multi-stage derivatives and discrete-time optimal stopping problems.

Several directions for future work are natural. First, it would be desirable to relax some technical assumptions such as the Lipschitz continuity of the coefficients and to extend the analysis to more general drivers. Second, incorporating variance reduction techniques, control variates, or problem-adapted network architectures may further reduce training time and improve stability in very high dimensions. Third, extending the framework to models with jumps, stochastic volatility, or path-dependent payoffs would broaden the range of applications and further test the flexibility of the compounding approach. Overall, the compound BSDE method offers a flexible and theoretically grounded route to forward deep-learning solvers for complex pricing and optimal stopping problems.

\section*{Acknowledgments}
The first author gratefully acknowledges financial support from the China Scholarship Council (CSC) through a PhD scholarship.

\bibliographystyle{unsrt} 
\bibliography{ref} 


\end{document}